\documentclass[orivec]{llncs}

\usepackage{hyperref}
\usepackage{eso-pic}
\AddToShipoutPicture*{
\put(133,110){
\scriptsize
Proceedings of \href{http://arxiv.org/abs/1406.1510}{$11^\text{th}$ Workshop on Constraint Handling Rules (CHR 2014)}, 
paper CHR/2014/3.
}}

\usepackage{mathtools}
\usepackage{amsmath}
\usepackage{cdsty}
\usepackage{stmaryrd}
\usepackage{amssymb}
\usepackage{hyperref}

\newcommand{\sgap}{\quad}
\newcommand{\bgap}{\quad\quad}

\newcommand{\atsign}{@}

\newcommand{\myirule}[3][]
{{\renewcommand{\arraystretch}{1.2}
  \begin{array}{c} #2
  \\ \hline #3
  \end{array}}\;{\scriptstyle #1}
}
\def\ruleform#1{{\setlength{\fboxrule}{0.1pt}\fbox{\normalsize $#1$}}}

\newcommand{\tx}[1]{\texttt{#1}}

\newcommand{\mmid}{~\mid~}

\newcommand{\fig}[3]
        {\begin{figure*}[t]#3\
        \caption{\label{#1}#2}\ 
        \vspace{-5mm}
         \end{figure*} }

\newcommand{\prog}{{\cal P}}

\newcommand{\mt}[1]{\mathit{#1}}

\newcommand{\mr}[1]{\mathrm{#1}}

\newcommand{\subst}[2]{[#2/#1]}

\newcommand{\freevars}[1]{\mt{FV}(#1)}

\newcommand{\comment}[1]{}

\newcommand{\guardjudge}[1]{\models #1}
\newcommand{\notguardjudge}[1]{\not \models #1}

\newcommand{\mset}[1]{{\lbag #1 \rbag}}
\newcommand{\msetcomma}{,}

\newcommand{\msetempty}{\emptyset}

\newcommand{\termalpha}{t_{\alpha}}
\newcommand{\guardalpha}{g_{\alpha}}

\newcommand{\mirule}[1]{{\bf l}_{\mt{#1}}}
\newcommand{\rmirule}[1]{{\bf l}^{\neg}_{\mt{#1}}}
\newcommand{\ruirule}[1]{{\bf u}^{\neg}_{\mt{#1}}}

\newcommand{\rhsirule}[1]{{\bf r}_{\mt{#1}}}

\newcommand{\mysection}[1]{\vspace*{-3mm}\section{#1}\vspace*{-2mm}}

\newcommand{\andcompre}[2]{\bigwedge_{#1} \lbag #2 \rbag}
\newcommand{\mcompre}[3]{{\lbag #1 ~\tx{|}~ #2 \rbag_{#3}}}
\newcommand{\ngmcompre}[2]{{\lbag #1 \rbag_{#2}}}
\newcommand{\unidis}{\in}

\newcommand{\chrrule}[5]{ #1 ~@~ #2 ~\backslash~ #3 \Longleftrightarrow #4 ~\tx{|}~ #5 }

\newcommand{\abstrans}{\mapsto_{\alpha}}
\newcommand{\abstransstar}{\mapsto^*_{\alpha}}

\newcommand{\optrans}{\mapsto_{\omega}}

\newcommand{\subsume}[2]{#1 \sqsubseteq_{{\bf lhs}} #2}
\newcommand{\notsubsume}[2]{#1 \not \sqsubseteq_{{\bf lhs}} #2}

\newcommand{\unfsubsume}[3]{#1 \rhd #2 \sqsubseteq_{{\bf unf}} #3}
\newcommand{\notunfsubsume}[3]{#1 \rhd #2 \not \sqsubseteq_{{\bf unf}} #3}
\newcommand{\existunfjudge}[3]{#1 \rhd #2 \triangleq^{\neg}_{{\bf unf}} #3}

\newcommand{\existunfprogjudge}[2]{#1 \triangleq^{\neg}_{{\bf unf}} #2}
\newcommand{\notexistunfprogjudge}[2]{#1 \not \triangleq^{\neg}_{{\bf unf}} #2}

\newcommand{\initgoal}{\tx{init}}
\newcommand{\eagergoal}{\tx{eager}}
\newcommand{\lazygoal}{\tx{lazy}}

\newcommand{\actgoal}{\tx{act}}
\newcommand{\propgoal}{\tx{prop}}

\newcommand{\stackempty}{\epsilon}
\newcommand{\stack}[2]{\tx{[}#1~\tx{|}~#2\tx{]}}
\newcommand{\sstack}[1]{\tx{[}#1\tx{]}}
\newcommand{\stplus}{+}

\newcommand{\st}{\mt{St}}
\newcommand{\gs}{\mt{Gs}}
\newcommand{\lst}{\mt{Ls}}

\newcommand{\droplabels}[1]{\mt{dropLabels}(#1)}
\newcommand{\getlabels}[1]{\mt{getLabels}(#1)}
\newcommand{\newlabels}[2]{\mt{newLabels}(#1,#2)}

\newcommand{\execstate}[2]{{\langle#1~;~#2\rangle}}
\newcommand{\idcons}[2]{#1 \# #2}

\newcommand{\matchjudge}[2]{#1 \triangleq_{{\bf lhs}} #2}
\newcommand{\existmatchjudge}[2]{#1 \triangleq^{\neg}_{{\bf lhs}} #2}
\newcommand{\maximatchjudge}[2]

\newcommand{\rhsjudge}[2]{#1 \ggg_{{\bf rhs}} #2}

\newcommand{\drop}[1]{#1}
\newcommand{\dropidx}[1]{\mt{dropIdx}(#1)}
\newcommand{\getidx}[1]{\mt{getIdx}(#1)}

\newcommand{\lkprog}[2]{#1[#2]}

\newcommand{\lzgs}[1]{\mt{lazyGs}(#1)}
\newcommand{\eggs}[1]{\mt{eagerGs}(#1)}

\newcommand{\occ}[2]{#1:#2}

\newcommand{\pred}[2]{#1\tx{(}#2\tx{)}}

\newcommand{\corresp}[1]{\llceil #1 \rrceil}

\newcommand{\chrcp}{{\mt{CHR}^\mt{cp}}}

\newcommand{\reduce}{{\cal R}}

\newcommand{\ruleabstrans}[3]{ #1 \rhd #2 ~\abstrans~ #3 }
\newcommand{\progabstrans}[3]{ #1 \rhd #2 ~\abstrans^*~ #3 }

\newcommand{\ruleoptrans}[3]{ #1 \rhd #2 ~\optrans~ #3 }
\newcommand{\progoptrans}[3]{ #1 \rhd #2 ~\optrans^*~ #3 }

\newcommand{\wrule}{R_\omega}
\newcommand{\wprog}{{\cal P}_\omega}

\comment{
\newtheorem{theorem}{Theorem}
\newtheorem{lemma}[theorem]{Lemma}

\newenvironment{proof}{{\bf Proof:}~}{$\Box$ \\ }

}

\title{Constraint Handling Rules with Multiset Comprehension Patterns%
\thanks{This paper was made possible by
      grant NPRP 09-667-1-100, \emph{Effective Programming for Large
      Distributed Ensembles}, from the Qatar National Research Fund (a member
      of the Qatar Foundation). The statements made herein are solely the
      responsibility of the authors.}
}
\author{Edmund S. L. Lam and Iliano Cervesato}

\institute{
 Carnegie Mellon University \\
 \email{sllam@qatar.cmu.edu} and
 \email{iliano@cmu.edu} }

\begin{document}
\maketitle

\bibliographystyle{plain}

\begin{abstract}
  CHR is a declarative, concurrent and committed choice rule-based constraint
  programming language.  We extend CHR with multiset comprehension patterns,
  providing the programmer with the ability to write multiset rewriting rules
  that can match a variable number of constraints in the store. This enables
  writing more readable, concise and declarative code for algorithms that
  coordinate large amounts of data or require aggregate operations. We call
  this extension $\chrcp$. We give a high-level abstract semantics of
  $\chrcp$, followed by a lower-level operational semantics. We then show the
  soundness of this operational semantics with respect to the abstract
  semantics.
\end{abstract}

\section{Introduction}
\label{sec:intro}
 
CHR is a declarative, concurrent and committed choice rule-based constraint
programming language.  CHR rules are executed in a pure forward-chaining
(data-driven) and committed choice (no backtracking) manner, providing the
programmer with a highly expressive programming model to implement complex
programs in a concise and declarative manner.  Yet, programming in a pure
forward-chaining model is not without its shortfalls. Expressive as it is,
when faced with algorithms that operate over a dynamic number of constraints
(e.g., finding the minimum value or finding \emph{all} constraints in the
store matching a particular pattern), a programmer is forced to decompose
his/her code over several rules, as a CHR rule can only match a fixed number
of constraints. Such an approach is tedious, error-prone and leads to repeated
instances of boilerplate codes, suggesting the opportunity for a higher form
of abstraction.

This paper explores an extension of CHR with \emph{multiset comprehension
  patterns}~\cite{cruz-damp12-short,DBLP:conf/iclp/SneyersWSD07}.  These
patterns allow the programmer to write multiset rewriting rules that can match
dynamically-sized constraint sets in the store.  They enable writing more
readable, concise and declarative programs that coordinate large amount of
data or use aggregate operations. We call this extension $\chrcp$.

While defining an abstract semantics that accounts for comprehension patterns
is relatively easy, turning it into an efficient model of computation akin to
the refined operational semantics of CHR~\cite{Duck04therefined} is
challenging.  The problem is that
monotonicity~\cite{DBLP:journals/lncs/Fruhwirth94}, a key requirement for the
kind of incremental processing that underlies CHR's refined operational
semantics, does not hold in the presence of comprehension patterns.  We
address this issue by statically identifying $\chrcp$ constraints that are
monotonic, and limiting incremental processing to just these constraints.
Similarly to~\cite{Duck04therefined}, this approach yields a sound
transformation of the abstract model of computation into an implementable
system.

Altogether, this paper makes the following contributions:
\begin{itemize}
\item%
  We formally define the abstract syntax and abstract semantics of $\chrcp$.
\item%
  We define a notion of conditional monotonicity for $\chrcp$ programs, and
  define an operational semantics that exploits it to drive an efficient
  execution model for $\chrcp$.
\item%
  We prove the soundness of this operational semantics with respect to the
  abstract semantics.
\end{itemize}
The rest of the paper is organized as follows: Section~\ref{sec:eg} introduces
$\chrcp$ by examples and Section~\ref{sec:syntax} formalizes its syntax.
Section~\ref{sec:abs_sem} defines the abstract semantics while
Section~\ref{sec:cond_mono} examines monotonicity in $\chrcp$. In
Section~\ref{sec:op_sem}, we introduce an operational semantics for $\chrcp$
and in Section~\ref{sec:sound}, we prove its soundness with respect to the
abstract semantics.  Section~\ref{sec:related} situates $\chrcp$ in the
literature and Section~\ref{sec:conclude} outlines directions of future work.

\section{Motivating Examples}
\label{sec:eg}

In this section, we illustrate the benefits of comprehension patterns in
multiset rewriting on some examples.  A comprehension pattern 
$\mcompre{\pred{p}{\vec{t}}}{g}{\vec{x} \unidis t}$ represents a multiset 
of constraints that match the atomic constraint $\pred{p}{\vec{t}}$ and 
satisfy guard $g$ under the bindings of variables $\vec{x}$ that range over $t$,
known as the {\em comprehension domain}.

Consider the problem of swapping data among agents based on a pivot value.  We
express an integer datum $D$ belonging to agent $X$ by the constraint
$\mt{data(X,D)}$.  Then, given agents $X$ and $Y$ and pivot value $P$, we want
all of $X$'s data with value greater than or equal to $P$ to be transferred to
$Y$ and all of $Y$'s data less than $P$ to be transferred to $X$.  The
following $\chrcp$ rule implements this pivot swap procedure:%
{\small
\[
\begin{array}{l}
  \mt{pivotSwap} ~\atsign~
  \begin{array}{l}
    \mt{swap(X,Y,P)} \\
    \mt{\mcompre{data(X,D)}{D \geq P}{D \unidis Xs}} \\
    \mt{\mcompre{data(Y,D)}{D < P}{D \unidis Ys}} 
  \end{array}
  ~\Longleftrightarrow~
  \begin{array}{l}
    \mt{\ngmcompre{data(Y,D)}{D \unidis Xs}} \\
    \mt{\ngmcompre{data(X,D)}{D \unidis Ys}} 
  \end{array}
\end{array}
\]}%
The swap is triggered by the constraint $\mt{swap(X,Y,P)}$.  All of $X$'s data
that are greater than or equal to the pivot $P$ are identified by the
comprehension pattern $\mt{\mcompre{data(X,D)}{D \geq P}{D \unidis Xs}}$.
Similarly, all $Y$'s data less than $P$ are identified by
$\mt{\mcompre{data(Y,D)}{D < P}{D \unidis Ys}}$.  The instances of $D$ matched
by each comprehension pattern are accumulated in the comprehension domains
$\mt{Xs}$ and $\mt{Ys}$, respectively.  Finally, these collected bindings are
used in the rule body to complete the rewriting by redistributing all of $X$'s
selected data to $Y$ and vice versa.  The comprehension domains $\mt{Xs}$ and
$\mt{Ys}$ are treated as output variables in the rule head, since the
matches for $D$ are fetched from the store.  In the rule body, comprehension
ranges are input variables, as we construct the desired multisets of
constraints from them.  The $\chrcp$ semantics enforces the property that each
comprehension pattern captures a \emph{maximal multiset} of constraints in the
store, thus guaranteeing that no data that is to be swapped is left behind.

Comprehension patterns allow the programmer to easily write rule patterns that
manipulate dynamic numbers of constraints.  Now consider how the above program
would be written in pure CHR (without comprehension patterns).  To do this,
we are forced to explicitly implement the operation of collecting a multiset
of $\mt{data}$ constraints over several rules. We also need to introduce an
accumulator to store bindings for the matched facts as we retrieve them. A
possible implementation of this nature is as follows:%
{\small
$$
\begin{array}{@{}r@{\;\atsign\;}l@{\;\Longleftrightarrow\;}l@{}}
   \mt{init}
 & \mt{swap(X,Y,P)}
 & \mt{grabGE(X,P,Y,[\,]), grabLT(Y,P,X,[\,])} 
\\[0.5ex]
   \mt{ge1}
 & \mt{grabGE(X,P,Y,Ds), data(X,D)}
 & D \geq P \mid \mt{grabGE(X,P,Y,[D\mid Ds])}
\\ \mt{ge2}
 & \mt{grabGE(X,P,Y,Ds)}
 & \mt{unrollData(Y,Ds)}
\\[0.5ex]
   \mt{lt1}
 & \mt{grabLT(Y,P,X,Ds), data(Y,D)}
 & D < P \mid \mt{grabLT(Y,P,X,[D\mid Ds])}
\\ \mt{lt2}
 & \mt{grabLT(Y,P,X,Ds)}
 & \mt{unrollData(X,Ds)} 
\\[0.5ex]
   \mt{unroll1}
 & \mt{unrollData(L,[D\mid Ds])}
 & \mt{unrollData(L,Ds),data(L,D)}
\\ \mt{unroll2}
 & \mt{unrollData(L,[\,])}
 & \mt{true}
\end{array}
$$
}%
Here, $[\,]$ denotes the empty list and $\mt{[D\mid Ds]}$ constructs a list with the 
head element $D$ and the rest from $\mt{Ds}$.  In a CHR program that consists of several subroutines of
this nature, this boilerplate code gets repeated over and over, making the
program less concise. Furthermore, the use of list accumulators and auxiliary
constraints (e.g., $\mt{grabGE}$, $\mt{unrollData}$) makes the implementation
less readable and more error-prone.  Most importantly, the swap operation as
written in $\chrcp$ is \emph{atomic} while the above CHR code involves many
rewrites, which could be interspersed by applications of other rules that
operate on $\mt{data}$ constraints.

Comprehension patterns also promote a concise way of coding term-level
aggregate computations: using a comprehension pattern's ability to retrieve a
dynamic number of constraints, we can compute aggregates with term-level map
and reduce operations over multisets of terms.  Consider the following
$\chrcp$ rule:%
{\small
\[
\begin{array}{@{}l@{}}
  \mt{removeNonMin} ~\atsign~
\\
\hspace*{1em}
  \begin{array}{l}
     \mt{remove(Gs),}~
     \mt{\mcompre{edge(X,Y,W)}{X \in Gs}{(X,Y,W) \unidis Es}}
  \end{array} \\[0.5ex]
  \bgap \Longleftrightarrow
  \begin{array}{l|}
    \mt{Es \neq \msetempty} \\ 
    \mt{Ws = \ngmcompre{W}{(X,Y,W) \unidis Es}} \\ 
    \mt{W_m = \reduce~\mr{min}~\infty~Ws} \\
    \mt{Rs = \mcompre{(X,Y,W)}{W_m < W}{(X,Y,W) \unidis Es}}~
  \end{array}
  ~
  \mt{\ngmcompre{edge(X,Y,W)}{(X,Y,W) \unidis Rs}}
\\\\[-1.5ex]
  \mr{where} 
  \sgap \mr{min} = \lambda x.~ \lambda y.~ 
                   \mr{if}~x \leq y ~\mr{then}~ x ~\mr{else}~ y
\end{array}
\]
}%
This $\chrcp$ rule identifies the minimum weight $W_m$ from a group $\mt{Gs}$
of edges in a directed graph and deletes all edges in that group with weight
$W_m$. Note that there could be several such minimal edges. We represent an
edge of weight $W$ between nodes $X$ and $Y$ with the constraint
$\mt{edge(X,Y,W)}$.  The fact $\mt{remove(Gs)}$ identifies the group $\mt{Gs}$
whose outgoing edges are the subject of the removal.  The minimum weight $W_m$
is computed by collecting all edges with origin in a node in $\mt{Gs}$
(constraint $\mt{\mcompre{edge(X,Y,W)}{X \in Gs}{(X,Y,W) \unidis Es}}$),
extracting their weight into the multiset $\mt{Ws}$ (with $\mt{Ws =
  \ngmcompre{W}{(X,Y,W) \unidis Es}}$) and folding the binary function
$\mr{min}$ over all of $\mt{Ws}$ by means of the term-level \emph{reduce}
operator $\reduce$ (constraint $\mt{W_m = \reduce~\mr{min}~\infty~Ws}$).  The
multiset $\mt{Rs}$ collects the edges with weight strictly greater than $W_m$
(constraint $\mt{Rs = \mcompre{(X,Y,W)}{W_m < W}{(X,Y,W) \unidis Es}}$).

\section{Syntax}
\label{sec:syntax}

\fig{fig:syntax}{Abstract Syntax of $\chrcp$}{
$$
\begin{array}{c}
      \text{Variables:}~x
\bgap \text{Values:}~v
\bgap \text{Predicates:}~p
\bgap \text{Rule names:}~r
\\    \text{Primitive terms:}~\termalpha
\bgap \text{Primitive guards:}~\guardalpha
\end{array}
$$
$$
\begin{array}{r@{:\hspace{1em}}r@{\;\;\;::=\;\;\;}l}
   \text{\em Terms}
 & t
 & \termalpha \mmid \bar{t} \mmid \mcompre{t}{g}{\vec{x} \unidis t}
\\ \text{\em Guards}
 & g
& \guardalpha \mmid g \wedge g \mmid \andcompre{\vec{x} \unidis t}{g}
\\[1ex]
  \text{\em Atomic Constraints}
 & A
 & \pred{p}{~\vec{t}~}
\\ \text{\em Comprehensions}
 & M
 & \mcompre{A}{g}{\vec{x} \unidis t}
\\ \text{\em Rule Constraints}
 & C,B
 & A \mmid M
\\[1ex]
   \text{\em Rules}
 & R
 & \chrrule{r}{\bar{C}}{\bar{C}}{g}{\bar{C}}
\\ \text{\em Programs}
 & \prog
 & \bar{R}
  \end{array}
$$
\vspace{-3mm}
}

In this section, we define the abstract syntax of $\chrcp$.  We focus on the
core fragment of the $\chrcp$ language, on top of which convenient short-hands
and a ``sugared'' concrete syntax can be built.

Figure~\ref{fig:syntax} defines the abstract syntax of $\chrcp$. Throughout
this paper, we write $\bar{o}$ for a multiset of syntactic object $o$, with
$\varnothing$ indicating the empty multiset.  We write $\mset{\bar{o}_1
  \msetcomma \bar{o}_2}$ for the union of multisets $\bar{o}_1$ and
$\bar{o}_2$, omitting the brackets when no ambiguity arises.  The extension of
multiset $\bar{o}$ with syntactic object $o$ is similarly denoted
$\mset{\bar{o} \msetcomma o}$.  We write $\vec{o}$ for a comma-separated tuple
of $o$'s.

An atomic constraint $\pred{p}{\vec{t}}$ is a predicate symbol $p$ applied to
a tuple $\vec{t}$ of terms.  A comprehension pattern $\mcompre{A}{g}{\vec{x}
  \unidis t}$ represents a multiset of constraints that match the atomic
constraint $A$ and satisfy guard $g$ under the bindings of variables $\vec{x}$ that
range over $t$.  We call $\vec{x}$ the \emph{binding variables} and $t$ the
\emph{comprehension domain}.  The conjunctive comprehension of a multiset of
guards of the form $g$ is denoted by $\andcompre{\vec{x} \unidis t}{g}$. It
represents a conjunction of all instances of guard $g$ under the bindings of
$\vec{x}$ ranging over $t$.  In both forms of comprehension, the variables
$\vec{x}$ are locally bound with scope $g$ (and $A$).

The development of $\chrcp$ is largely agnostic with respect to the language
of terms.  We will assume a base term language $\mathcal{L}_\alpha$, that in
examples contains numbers and functions, but may be far richer.  We write
$\termalpha$ for a generic term in this base language, $\guardalpha$ for an
atomic guard over such terms, and $\models_\alpha$ for the satisfiability
relation over ground guards.  In addition to $\mathcal{L}_\alpha$, $\chrcp$
contains tuples with their standard operators, and a term-level multisets.
Multiset constructors include the empty multiset $\varnothing$, singleton
multisets $\mset{t}$ for any term $t$, and multiset union $\mset{m_1
  \msetcomma m_2}$ for multisets $m_1$ and $m_2$.  Term-level multiset
comprehension $\mcompre{t}{g}{x \unidis m}$ filters multiset $m$ according to
$g$ and maps the result as specified by $t$.  The reduce operator
$\reduce~f~e~m$ recursively combines the elements of multiset $m$ pairwise
according to $f$, returning $e$ for the empty multiset.

As in CHR, a $\chrcp$ rule $\chrrule{r}{\bar{C}_p}{\bar{C}_s}{g}{\bar{B}}$
specifies the rewriting of $\bar{C}_s$ into $\bar{B}$ under the conditions
that constraints $\bar{C}_p$ are available and guards $g$ are satisfied.  As
usual, we refer to $\bar{C}_p$ as the \emph{propagated head}, to $\bar{C}_s$
as the \emph{simplified head} and to $\bar{B}$ as the \emph{body} of the rule.
If the propagated head $\bar{C}_p$ is empty or the guard $g$ is always
satisfied (i.e., $\mt{true}$), we omit the respective rule component entirely.
Rules with an empty simplified head $\bar{C}_s$ are referred to as
\emph{propagation} rules.  All free variables in a $\chrcp$ rule are
implicitly universally quantified at the head of the rule.  We will assume
that a rule's body is grounded by the rule heads and that guards (built-in
constraint) cannot appear in the rule body. This simplifies the discussion,
allowing us to focus on the novelties brought about by comprehension patterns.

\section{Abstract Semantics}
\label{sec:abs_sem}

This section describes the abstract semantics of $\chrcp$. We first define
some meta-notation and operations.  The set of the free variables in a
syntactic object $o$ is denoted $\freevars{o}$.  We write
$\subst{\vec{x}}{\vec{t}}o$ for the simultaneous replacement within object $o$
of all occurrences of variable $x_i$ in $\vec{x}$ with the corresponding term
$t_i$ in $\vec{t}$.  When traversing a binding construct (e.g., comprehension
patterns), substitution implicitly $\alpha$-renames variables to avoid
capture.  It will be convenient to assume that terms get normalized during (or
right after) substitution.

Without loss of generality, we assume that atomic constraints in a $\chrcp$
rule have the form $p(\vec{x})$, including in comprehension patterns.  This
simplified form pushes complex term expressions and computations into the
guard component of the rule or the comprehension pattern.  The satisfiability
of a ground guard $g$ is modeled by the judgment $\guardjudge{g}$; its
negation is written $\notguardjudge{g}$.

The abstract semantics of $\chrcp$ is modeled by the small-step judgment
$\ruleabstrans{\prog}{\st}{\st'}$, which applies a rule in $\chrcp$ program
$\prog$ to constraint store $\st$ producing store $\st'$.  A constraint store
is a multiset of ground atomic constraints.  Applying a rule has two phases:
we match its heads and guard against the current store, and whenever
successful, we replace some of the matched facts with the corresponding instance
of this body.  We will now describe these two phases in isolation and then
come back to rule application.

\fig{fig:matching_sem}{Semantics of Matching in $\chrcp$}{
$\ruleform{\text{Matching:}
     \bgap \matchjudge{\bar{C}}{\st}
     \bgap \matchjudge{C}{\st}}$
$$
\begin{array}{c}
    \myirule[(\mirule{\mt{mset}\tx{-}1})]
      {\matchjudge{\bar{C}}{\st}
       \quad
       \matchjudge{C}{\st'}}
      {\matchjudge{\mset{\bar{C} \msetcomma C}}{\mset{\st \msetcomma \st'}}}
\qquad
    \myirule[(\mirule{\mt{mset}\tx{-}2})]
      {}
      {\matchjudge{\msetempty}{\msetempty}}
\qquad
    \myirule[(\mirule{atom})]
      {}
      {\matchjudge{A}{A}}
\\\\
    \myirule[(\mirule{comp\tx{-}1})]
      {\matchjudge{\subst{\vec{x}}{\vec{t}}A}{A'}
       \quad
       \guardjudge{\subst{\vec{x}}{\vec{t}}g}
       \quad
       \matchjudge{\mcompre{A}{g}{\vec{x} \unidis \mt{ts}}}{\st}}
     {\matchjudge{\mcompre{A}{g}{\vec{x} \unidis \mset{\mt{ts} \msetcomma \vec{t}}}}{\mset{\st \msetcomma A'}}}
\qquad
    \myirule[(\mirule{comp\tx{-}2})]
      {}
      {\matchjudge{\mcompre{A}{g}{\vec{x} \unidis \msetempty}}{\msetempty}}
\end{array}
$$
\medskip

$\ruleform{\text{Residual Non-matching:}
      \bgap \existmatchjudge{\bar{C}}{\st}
      \bgap \existmatchjudge{C}{\st}}$
$$
\begin{array}{c}
    \myirule[(\rmirule{mset\tx{-}1})]
      {\existmatchjudge{\bar{C}}{\st}
       \quad
       \existmatchjudge{C}{\st}}
      {\existmatchjudge{\mset{\bar{C} \msetcomma C}}{\st}}
\qquad
    \myirule[(\rmirule{mset\tx{-}2})]
      {}
      {\existmatchjudge{\msetempty}{\st}}
\\\\
    \myirule[(\rmirule{atom})]
      {}
      {\existmatchjudge{A}{\st}}
\qquad
    \myirule[(\rmirule{comp\tx{-}1})]
      {\notsubsume{A}{M}
       \quad
       \existmatchjudge{M}{\st}}
      {\existmatchjudge{M}{\mset{\st \msetcomma A}}}
\qquad
    \myirule[(\rmirule{comp\tx{-}2})]
      {}
      {\existmatchjudge{M}{\msetempty}}
\end{array}
$$
$\text{Subsumption:}~
  \subsume{A}{\mcompre{A'}{g}{\vec{x} \unidis \mt{ts}}} \sgap \text{iff} \sgap 
  A = \theta A'~\text{and}~\guardjudge{\theta g} ~\text{ for some}~\theta = \subst{\vec{x}}{\vec{t}}$
}

Figure~\ref{fig:matching_sem} defines the matching phase of $\chrcp$.  It
relies on two forms of judgments, each with a variant operating on a multiset
of constraint patterns $\bar{C}$ and a variant on an individual pattern $C$.
The first matching judgment, $\matchjudge{\bar{C}}{\st}$, holds when the
constraints in the store fragment $\st$ match \emph{completely} the multiset
of constraint patterns $\bar{C}$.  It will always be the case that $\bar{C}$
is closed (i.e., $\freevars{\bar{C}} = \emptyset$).  Rules
$(\mirule{mset\tx{-}*})$ iterate rules $(\mirule{atom})$ and
$(\mirule{comp\tx{-}*})$ on $\st$, thereby partitioning it into fragments
matched by these rules.  Rule $(\mirule{atom})$ matches an atomic constraint
$A$ to the singleton store $A$.  Rules $(\mirule{comp\tx{-}*})$ match a
comprehension pattern $\mcompre{A}{g}{\vec{x} \unidis \mt{ts}}$. If the
comprehension domain is empty ($x \unidis \msetempty$), the store must be empty
(rule $\mirule{comp\tx{-}2}$). Otherwise, rule $(\mirule{comp\tx{-}1})$ binds
$\vec{x}$ to an element $\vec{t}$ of the comprehension domain $\mt{ts}$,
matches the instance $\subst{\vec{x}}{\vec{t}}A$ of the pattern $A$ with a
constraint $A'$ in the store if the corresponding guard instance
$\subst{\vec{x}}{\vec{t}}g$ is satisfiable, and continues with the rest of the
comprehension domain.

To guarantee the maximality of comprehension patterns, we test a store for
\emph{residual matchings}.  This relies on the matching subsumption relation
$\subsume{A}{\mcompre{A'}{g}{\vec{x} \unidis \mt{ts}}}$, defined at the bottom
of Figure~\ref{fig:matching_sem}.  This relation holds if $A$ can be absorbed
into the comprehension pattern $\mcompre{A'}{g}{\vec{x} \unidis \mt{ts}}$.
Note that it ignores the available bindings in $\mt{ts}$: $t$ need not be an
element of the comprehension domain.  Its negation is denoted by
$\notsubsume{A}{\mcompre{A'}{g}{\vec{x} \unidis \mt{ts}}}$.  We test a store
for residual matchings using the \emph{residual non-matching judgment}
$\existmatchjudge{\bar{C}}{\st}$. Informally, for each comprehension pattern
$\mcompre{A'}{g}{\vec{x} \unidis \mt{ts}}$ in $\bar{C}$, this judgment checks
that no constraints in $\st$ matches $A'$ satisfying $g$.  This judgment is
defined in the middle section of Figure~\ref{fig:matching_sem}.  Rules
($\rmirule{mset\tx{-}*}$) apply the remaining rules to each constraint
patterns $C$ in $\bar{C}$.  Observe that each pattern $C$ is ultimately
matched against the entire store $\st$.  Rule ($\rmirule{atom}$) asserts that
atoms have no residual matches.  Rules ($\rmirule{comp\tx{-}*}$) check that no
constraints in $\st$ match the comprehension pattern $M =
\mcompre{A'}{g}{\vec{x} \unidis \mt{ts}}$.

\fig{fig:rule_app_judgment}{Abstract Semantics of $\chrcp$}{
$\ruleform{\text{Rule Body:}
     \sgap \rhsjudge{\bar{C}}{\st}
     \sgap \rhsjudge{C}{\st}}$
$$
\begin{array}{c}
    \myirule[(\rhsirule{mset\tx{-}1})]
      {\rhsjudge{\bar{C}}{\st}
       \quad
       \rhsjudge{C}{\st'}}{\rhsjudge{\mset{\bar{C} \msetcomma C}}
      {\mset{\st \msetcomma \st'}}}
\qquad
    \myirule[(\rhsirule{mset\tx{-}2})]
      {}
      {\rhsjudge{\msetempty}{\msetempty}}
\qquad
    \myirule[(\rhsirule{atom})]
      {}
      {\rhsjudge{A}{A}}
\\\\[-1ex]
    \myirule[(\rhsirule{comp\tx{-}1})]
      {\guardjudge{\subst{\vec{x}}{\vec{t}}g}
       \quad
       \rhsjudge{\subst{\vec{x}}{t}A}{A'}
       \quad
       \rhsjudge{\mcompre{A}{g}{\vec{x} \unidis \mt{ts}}}{A'} }
      {\rhsjudge{\mcompre{A}{g}{\vec{x} \unidis \mset{\mt{ts} \msetcomma \vec{t}}}}{\mset{\st \msetcomma A'}}}
\\\\[-1ex]
    \myirule[(\rhsirule{comp\tx{-}2})]
      {\notguardjudge{\subst{\vec{x}}{\vec{t}}g}
       \quad
       \rhsjudge{\mcompre{A}{g}{\vec{x} \unidis \mt{ts}}}{\st}}
      {\rhsjudge{\mcompre{A}{g}{\vec{x} \unidis \mset{\mt{ts} \msetcomma \vec{t}}}}{\st}}
\qquad
    \myirule[(\rhsirule{comp\tx{-}3})]
      {}
      {\rhsjudge{\mcompre{A}{g}{\vec{x} \unidis \msetempty}}{\msetempty}}
\end{array}
$$
\smallskip

$\ruleform{\text{Rule Application:}
     \prog \rhd \st \abstrans \st }$
$$
\begin{array}{c}
   \myirule
     {\begin{array}{c}
         (\chrrule{r}{\bar{C}_p}{\bar{C}_s}{g}{\bar{B}}) \in \prog
         \quad
         \guardjudge{\theta g}
      \\ \matchjudge{\theta\bar{C}_p}{\st_p}
         \quad \matchjudge{\theta\bar{C}_s}{\st_s}
         \quad \existmatchjudge{\theta \mset{\bar{C}_p \msetcomma \bar{C}_s}}{\st}
         \quad \rhsjudge{\theta\bar{B}}{\st_b}
      \end{array}}
    {\prog \rhd {\mset{\st_p \msetcomma \st_s \msetcomma \st}} \abstrans
             {\mset{\st_p \msetcomma \st_b \msetcomma \st}} }
\end{array}
$$
}

If an instance of a $\chrcp$ rule passes the matching phase, we need to
\emph{unfold} the comprehension patterns in its body into a multiset of atomic
constraints.  The judgment $\rhsjudge{\bar{C}}{\st}$, defined in
Figure~\ref{fig:rule_app_judgment}, does this unfolding.  This judgment is
similar to the matching judgment (Figure~\ref{fig:matching_sem}) except that
it skips any element in the comprehension domain that fails the guard (rule
$\rhsirule{comp\tx{-}2}$).

We now have all the pieces to define the application of a $\chrcp$ rule.  The
judgment $\ruleabstrans{\prog}{\st}{\st'}$ describes a state transition from
stores $\st$ to $\st'$ triggered by applying a rule instance in $\chrcp$
program $\prog$.  This judgment is defined by the rule at the bottom of
Figure~\ref{fig:rule_app_judgment}.  A $\chrcp$ rule
$\chrrule{r}{\bar{C}_p}{\bar{C}_s}{g}{\bar{B}} \in \prog$ is applicable in
$\st$ if there is a substitution $\theta$ that makes the guard satisfiable
(i.e., $\guardjudge{\theta g}$) and if there are fragments $\st_p$ and $\st_s$
of the store that match the corresponding instance of the propagated and
simplified heads ($\matchjudge{\theta\bar{C}_p}{\st_p}$ and
$\matchjudge{\theta\bar{C}_s}{\st_s}$) and that are maximal in $\st$ (i.e.,
$\existmatchjudge{\theta\mset{\bar{C}_p \msetcomma \bar{C}_s}}{\st}$).  We
then apply this rule by replacing the store fragment $\st_s$ that matches the
simplified head instance with the unfolded rule body instance
($\rhsjudge{\theta\bar{B}}{\st_b}$).  We write
$\progabstrans{\prog}{\st}{\st'}$ for zero to more applications of this rule.

\mysection{Monotonicity}
\label{sec:cond_mono}

In this section, we analyze the impact that comprehension patterns have on
monotonicity in $\chrcp$.  Specifically, we show that $\chrcp$ enjoys a
\emph{conditional} form of monotonicity, that we will exploit in
Section~\ref{sec:op_sem} to define an operational semantics for $\chrcp$ based
on (partial) incremental processing of constraints.

In CHR, monotonicity~\cite{DBLP:journals/lncs/Fruhwirth94} is an important
property.  Informally, monotonicity ensures that if a transition between two
CHR states (stores) is possible, it is also possible in any larger store. This
property underlies many efficient implementation techniques for CHR.\@ For
instance, the incremental processing of constraints in CHR's refined
operational semantics~\cite{Duck04therefined} is sound because of the
monotonicity property.  When parallelizing CHR
execution~\cite{DBLP:conf/ppdp/TriossiORF12}, the soundness of composing
concurrent rule application also depends on monotonicity.  In $\chrcp$
however, monotonicity is not guaranteed in its standard form:%
{\small
$$
\text{if}~\prog \rhd \st \abstrans \st', ~\text{then}~\prog \rhd
  \mset{\st,\st''} \abstrans \mset{\st',\st''} ~\text{for any}~\st''
$$}%
This is not surprising, since the maximality requirement of comprehension
patterns could be violated if we add a constraint $A \in \st''$.  Consider the
following example, where we extend the store with a constraint $\mt{a(3)}$
which can be matched by a comprehension pattern in program $\prog$:%
{\small
$$
\begin{array}{c}
  \mt{\prog ~\equiv~ r \atsign \ngmcompre{a(X)}{X \unidis Xs} \Longleftrightarrow \ngmcompre{b(X)}{X \unidis Xs}} \\ \\
  \mt{\prog \rhd \mset{a(1) \msetcomma a(2)} \abstrans \mset{b(1) \msetcomma b(2)}} ~\text{but}~ 
  \mt{\prog \rhd \mset{a(1) \msetcomma a(2) \msetcomma a(3)} \not \abstrans \mset{b(1) \msetcomma b(2) \msetcomma a(3)}}
\end{array}
$$}%
In this example, extending the store with $\mt{a(3)}$ violates the maximality
of comprehension pattern $\ngmcompre{a(X)}{X \unidis Xs}$. Hence, the
derivation under the larger store is not valid with respect to the abstract
semantics.  Yet all is not lost: if we can guarantee that $\st''$ only
contains constraints that \emph{never} match any comprehension pattern in the
head of any rule in $\prog$, we recover monotonicity, albeit in a restricted
form.  For instance, extending the store in the above example with constraint
$\mt{c(3)}$ does not violate monotonicity.

\fig{fig:resid_unif_sem}{Residual Non-Unifiability}{
$$
\begin{array}{@{}c@{}}
    \myirule[(\ruirule{prog\tx{-}1})]
      {\existunfjudge{g}{\mset{\bar{C}_p \msetcomma \bar{C}_s}}{\bar{B}}
       \quad
       \existunfprogjudge{\prog}{\bar{B}}}
      {\existunfprogjudge{\prog,(\chrrule{r}{\bar{C}_p}{\bar{C}_s}{g}{\bar{C}_b})}{\bar{B}}}
\qquad
    \myirule[(\ruirule{prog\tx{-}2})]
      {}
      {\existunfprogjudge{\msetempty}{\bar{B}}} ~ 
\\\\[-1ex]
    \myirule[(\ruirule{mset\tx{-}1})]
      {\existunfjudge{g}{\bar{C}}{\bar{B}}
       \quad
       \existunfjudge{g}{C}{\bar{B}}}
      {\existunfjudge{g}{\mset{\bar{C} \msetcomma C}}{\bar{B}}}
\quad
    \myirule[(\ruirule{mset\tx{-}2})]
      {}
      {\existunfjudge{g}{\msetempty}{\bar{B}}}
\quad
    \myirule[(\ruirule{atom})]
      {}
      {\existunfjudge{g}{A}{\bar{B}}} ~  
\\\\[-1ex]
    \myirule[(\ruirule{comp\tx{-}1})]
      {\notunfsubsume{g}{B}{M}
       \quad
       \existunfjudge{g}{M}{\bar{B}}}
      {\existunfjudge{g}{M}{\mset{\bar{B} \msetcomma B}}}
\qquad
    \myirule[(\ruirule{comp\tx{-}2})]
      {}
      {\existunfjudge{g}{M}{\msetempty}}
\\\\[-1ex]\hline 
\end{array}
$$
{\small
$\begin{array}{r@{\text{ \ iff \ }}l}
   \unfsubsume{g}{A}{\mcompre{A'}{g'}{\vec{x} \unidis \mt{ts}}}
 & \theta A \equiv \theta A',
   \guardjudge{\theta g'},
   \guardjudge{\theta g}
   \text{ for some } \theta
\\[1ex]
   \unfsubsume{g''}{\mcompre{A}{g}{\vec{x} \unidis \mt{ts}}}{\mcompre{A'}{g'}{\vec{x}' \unidis \mt{ts}'}}
 & \theta A \equiv \theta A',
   \guardjudge{\theta g''},
   \guardjudge{\theta g'},
   \guardjudge{\theta g}
   \text{ for some } \theta 
\end{array}$}
}

We formalize this idea by generalizing the residual non-matching judgment from
Figure~\ref{fig:matching_sem}.  The resulting \emph{residual non-unifiability
  judgment} is defined in Figure~\ref{fig:resid_unif_sem}. Given a program
$\prog$ and a multiset of constraint patterns $\bar{B}$, the judgment
$\existunfprogjudge{\prog}{\bar{B}}$ holds if no constraint that matches any
pattern in $\bar{B}$ can be unified with any comprehension pattern in any rule
heads of $\prog$. Rules $(\ruirule{prog\tx{-}*})$ iterate over each $\chrcp$
rule in $\prog$.  For each rule, the judgment
$\existunfjudge{g}{\bar{C}}{\bar{B}}$ tests each rule pattern in $\bar{C}$
against all the patterns $\bar{B}$ (rules $\ruirule{mset\tx{-}*}$).  Rule
($\ruirule{atom}$) handles atomic facts, which are valid by default.  Rules
($\ruirule{comp\tx{-}*}$) check that no body pattern $\bar{B}$ is unifiable
with any rule head pattern $\bar{C}$ under the guard $g$.  It does so on the
basis of the relations at the bottom of Figure~\ref{fig:resid_unif_sem}.

A constraint (atom or comprehension pattern) $C$ is \emph{monotone} w.r.t.\@
program $\prog$ if $\existunfprogjudge{\prog}{C}$ is derivable.
With the residual non-unifiability judgment, we can ensure the conditional monotonicity of
$\chrcp$.

\begin{theorem}[Conditional Monotonicity]
\label{theo:cond_mono}%
Given a program $\prog$ and stores $\st$ and $\st'$, if $\prog \rhd \st
\abstransstar \st'$, then for any store fragment $\st''$ such that
$\existunfprogjudge{\prog}{\st''}$, we have that $\prog \rhd \mset{\st
  \msetcomma \st''} \abstransstar \mset{\st' \msetcomma \st''}$.
\end{theorem}

\begin{proof}
  The proof proceeds by induction on the derivation $\prog \rhd \st
  \abstransstar \st'$.  The monotonicity property holds trivially in the base
  case where we make zero steps.  In the inductive case, we rely on the fact
  that if $\existunfprogjudge{\prog}{\st''}$, then, for any instance of a
  comprehension pattern $M$ occurring in a rule head $\prog$, we are
  guaranteed to have $\existmatchjudge{M}{\st''}$.
\qed%
\end{proof}

This theorem allows us to enlarge the context of any derivation $\prog \rhd
\st \abstransstar \st'$ with $\st''$, if we have the guarantee that all
constraints in $\st''$ are monotone with respect to $\prog$.

\mysection{Operational Semantics}
\label{sec:op_sem}

In this section, we define a lower-level operational semantics for $\chrcp$.
Similarly to~\cite{Duck04therefined}, this operational semantics determines a
goal-based execution of $\chrcp$ programs that utilizes monotonicity
(conditional, in our case) to incrementally process constraints. By
``incrementally'', we mean that goal constraints are added to the store one by
one, as we process each for potential match to the rule heads. The main
difference with~\cite{Duck04therefined} is that a goal constraint $C$ that is
not monotone w.r.t.\@ the program $\prog$ (i.e., such that
$\notexistunfprogjudge{\prog}{C}$) is stored immediately before any other rule
application is attempted. Similarly to other operational semantics for CHR, our
semantics also handles \emph{saturation}, enforcing the invariant that a
propagation rule instance is only applied once for each matching rule head
instance in the store. Hence, programs with propagation rules are not
necessarily non-terminating.  This makes our operational semantics incomplete
w.r.t.\@ the abstract semantics, but saturation is generally viewed as
desirable.

\fig{fig:op_sem_syntax}{Annotated Programs, Execution States and Auxiliary Meta-operations}{
\[
\hspace{-8mm}
\begin{array}{c}

 \begin{array}{ccc}

 \begin{array}{c}
  \text{Occurrence Index}~i \\ \text{Store Label}~n 
 \end{array}
 & \bgap &
 \begin{array}{llcl}
   \text{Rule Head} & H      & \sgap ::= \sgap & \occ{C}{i} \\
   \text{CHR Rule}  & \wrule & ::= & \chrrule{r}{\bar{H}}{\bar{H}}{g}{\bar{B}} \\
   \text{Program}   & \wprog & ::= & \bar{\wrule} 
 \end{array} 

 \end{array} \vspace{2mm} \\

 \begin{array}{lclcl}
   \text{Matching History}         & \sgap & \Theta  & ~::=~ & (\bar{\theta},\bar{n}) \\
   \text{Goal Constraint}          & \sgap & G       & ~::=~ & \initgoal~\bar{B} \mmid \lazygoal~A \mmid \eagergoal~\idcons{A}{n} \mmid \actgoal~\idcons{A}{n}~i 
                                                    \mmid \propgoal~\idcons{A}{n}~i~\Theta \\
   \text{Goal Stack}               & \sgap & \gs     & ~::=~ & \stackempty \mmid \stack{G}{\gs} \\
   \text{Labeled Store}         & \sgap & \lst    & ~::=~ & \msetempty  \mmid \mset{\lst \msetcomma \idcons{A}{n}} \\
   \text{Execution State}      & \sgap & \sigma  & ~::=~ & \execstate{\gs}{\lst}
 \end{array} \vspace{2mm} \\

 \begin{array}{ccccccc}
   \dropidx{\occ{C}{i}} \Coloneqq C & \sgap & \getidx{\occ{C}{i}} \Coloneqq \{i\} & \sgap & \droplabels{\idcons{A}{n}} \Coloneqq A 
   & \sgap & \getlabels{\idcons{A}{n}} \Coloneqq \{n\}
 \end{array} \vspace{2mm} \\

 \begin{array}{c}
   \newlabels{\lst}{A} \Coloneqq \idcons{A}{n} \sgap \text{such that}~n \notin \getlabels{\lst}
 \end{array} \vspace{2mm} \\

 \begin{array}{cc}
   \lkprog{\wprog}{i} \sgap ::= &
   \begin{cases}
     \sgap \wrule & \text{if}~ \wrule \in \wprog ~\text{and}~ i \in \getidx{\wrule} \\
     \sgap \bot   & \text{otherwise}
   \end{cases}
 \end{array} \vspace{2mm} \\

 \begin{array}{ccccc}
   \myirule{\matchjudge{\dropidx{\bar{H}}}{ \droplabels{\lst} }}{\matchjudge{\bar{H}}{\lst}}
   & \bgap &
   \myirule{\existmatchjudge{\dropidx{\bar{H}}}{\droplabels{\lst}}}{\existmatchjudge{\bar{H}}{\lst}}
   & \bgap &
   \myirule{\existunfprogjudge{\dropidx{\prog}}{\bar{C}}}{\existunfprogjudge{\prog}{\bar{C}}}
 \end{array}

\comment{
 \begin{array}{rl}
   \text{Unfolding Goals} \bgap &
   \begin{cases}
     \sgap \eggs{\mset{\lst \msetcomma \idcons{A}{n}}}   & ::= \bgap \stack{\eagergoal~\idcons{A}{n}}{\eggs{\lst}} \\
     \sgap \eggs{\msetempty}                & ::= \bgap \epsilon \\
     \sgap \lzgs{\mset{\st_m \msetcomma A}} & ::= \bgap \stack{\lazygoal~A}{\lzgs{\st_m}} \\
     \sgap \lzgs{\msetempty}                & ::= \bgap \epsilon
   \end{cases}
 \end{array}
}
\end{array}
\]
\vspace{-6mm}
}

Figure~\ref{fig:op_sem_syntax} defines the execution states of $\chrcp$
programs in this operational semantics and some auxiliary notions.  We
annotate a $\chrcp$ program $\prog$ with \emph{rule head occurrence indices}.
The result is denoted $\wprog$.  Specifically, each rule head pattern $C$ of
$\prog$ is annotated with a unique integer $i$ starting from $1$, and is
written $\occ{C}{i}$ in $\wprog$. This represents the order in which rule
heads are matched against a constraint.  Execution states are pairs $\sigma =
\execstate{\gs}{\lst}$ where $\gs$ is the \emph{goal stack} and $\lst$ is the
\emph{labeled store}. The latter is a constraint store with each constraint
annotated with a unique label $n$.  This label allows us to distinguish
between copies of the same constraint in the store and to uniquely associate a
goal constraint with a specific stored constraint.  Labels also support
saturation for propagation rules (see below).  Each goal in a goal stack $\gs$
represents a unit of execution and $\gs$ itself is a sequence of goals to be
executed.  A non-empty goal stack has the form $\stack{G}{\gs}$, where $G$ is
the goal at the top of the stack and $\gs$ the rest of the stack. The empty
stack is denoted $\epsilon$. We abbreviate the singleton stack containing $G$
as $\sstack{G}$.  Given two stacks $\gs_1$ and $\gs_2$ we denote their
concatenation as $\gs_1 \stplus \gs_2$. We write $G \in \gs$ to denote that
$G$ occurs in some position of $\gs$.  Unlike~\cite{Duck04therefined}, we
attach a label to each goal.  These labels are $\initgoal$, $\lazygoal$,
$\eagergoal$, $\actgoal$ and $\propgoal$.  We will explain the purpose of each
as we describe the semantics.

Figure~\ref{fig:op_sem_syntax} defines several auxiliary operations that
either retrieve or drop occurrence indices and store labels: $\dropidx{H}$ and
$\getidx{H}$ deal with indices, $\droplabels{\_}$ and $\getlabels{\_}$ with
labels.  We inductively extend $\dropidx{\_}$ to multisets of rule heads and
annotated programs, each returning the respective syntactic construct with
occurrence indices removed.  Likewise, we extend $\getidx{\_}$ to multisets of
rule heads and $\chrcp$ rules, to return the set of all occurrence indices
that appear in them. We similarly extend $\droplabels{\_}$ and
$\getlabels{\_}$ to be applicable with labeled stores.  As a means of
generating new labels, we also define the operation $\newlabels{\lst}{A}$ that
returns $\idcons{A}{n}$ such that $n$ does not occur in $\lst$. Given
annotated program $\wprog$ and occurrence index $i$, $\lkprog{\wprog}{i}$
denotes the rule $\wrule \in \wprog$ in which $i$ occurs, or $\bot$ if $i$ does
not occur in any of $\wprog$'s rules.  The bottom of
Figure~\ref{fig:op_sem_syntax} also defines extensions to the match, residual
non-matching and residual unifiability judgment to annotated entities. Applied
to the respective occurrence indexed or labeled syntactic constructs, these
judgments simply strip away occurrence indices or labels.

\fig{fig:op_sem_core}{Operational Semantics of $\chrcp$ (Core-Set)}{
{\small
\[
\hspace{-10mm}
\arraycolsep=1.4pt\def\arraystretch{1.2}
\begin{array}{|c|l|}
   \hline
   (\mt{init}) & \begin{array}{l}
                    \wprog \rhd \execstate{\stack{\initgoal~\mset{\bar{B}_l \msetcomma \bar{B}_e}}{\gs}}{\lst}
                         \optrans \execstate{\lzgs{\st_l} \stplus \eggs{\lst_e} \stplus \gs}{\mset{\lst \msetcomma \lst_e}} \\
                    \text{such that}~\existunfprogjudge{\wprog}{\bar{B}_l} \sgap \rhsjudge{\bar{B}_e}{\st_e} \sgap 
                         \rhsjudge{\bar{B}_l}{\st_l} \sgap \lst_e = \newlabels{\lst}{\st_e} \\
                    \begin{array}{lll}
                      \text{where}~& \eggs{\mset{\lst \msetcomma \idcons{A}{n}}} \Coloneqq \stack{\eagergoal~\idcons{A}{n}}{\eggs{\lst}} 
                                     \bgap &  \eggs{\msetempty} \Coloneqq \epsilon \\
                                   & \lzgs{\mset{\st_m \msetcomma A}} \Coloneqq \stack{\lazygoal~A}{\lzgs{\st_m}} 
                                           &  \lzgs{\msetempty}          \Coloneqq \epsilon
                    \end{array}
                 \end{array} \\ \hline 

   (\mt{lazy}\tx{-}\mt{act})    &  \begin{array}{l}
                                     \wprog \rhd \execstate{\stack{\lazygoal~A}{\gs}}{\lst}  
                                         \optrans \execstate{\stack{\actgoal~\idcons{A}{n}~1}{\gs}}{\mset{\lst \msetcomma \idcons{A}{n}}}  \\
                                     \text{such that}~ {\mset{\idcons{A}{n}}} = \newlabels{\lst}{\mset{A}}
                                  \end{array} \\ \hline

   (\mt{eager}\tx{-}\mt{act}) &  \wprog \rhd \execstate{\stack{\eagergoal~\idcons{A}{n}}{\gs}}{\mset{\lst \msetcomma \idcons{A}{n}}}  
                                         \optrans \execstate{\stack{\actgoal~\idcons{A}{n}~1}{\gs}}{\mset{\lst \msetcomma \idcons{A}{n}}}  \\ \hline

   (\mt{eager}\tx{-}\mt{drop}) &  \wprog \rhd \execstate{\stack{\eagergoal~\idcons{A}{n}}{\gs}}{\lst}  
                                    \optrans \execstate{\gs}{\lst} \bgap \bgap \text{if}~ \idcons{A}{n} \notin \lst \\ \hline

   (\mt{act}\tx{-}\mt{simpa}\tx{-}1) &  \begin{array}{l}
                         
                         \begin{array}{l}
                           \wprog ~ \rhd \execstate{\stack{\actgoal~\idcons{A}{n}~i}{\gs}}
                                                    {\mset{\lst \msetcomma \lst_p \msetcomma \lst_s \msetcomma \lst_a \msetcomma \idcons{A}{n}}}
                           \optrans \execstate{\stack{\initgoal~\theta\bar{B}}{\gs}}{\mset{\lst \msetcomma \lst_p}} 
                         \end{array} \\
                          \text{if}~\lkprog{\wprog}{i} = (\chrrule{r}{\bar{H}_p}{\mset{\bar{H}_s \msetcomma \occ{C}{i}}}{g}{\bar{B}}),~\text{there exists some}~\theta~\text{such that}~ \\
                         \sgap \begin{array}{lll}
                           - & \guardjudge{\theta g} \bgap \matchjudge{\theta C}{\mset{\drop{\lst_a} \msetcomma \idcons{A}{n}}} & (\text{Guard Satisfied and Active Match}) \\
                           - & \matchjudge{\theta\bar{H}_p}{\drop{\lst_p}} \bgap \matchjudge{\theta\bar{H}_s}{\drop{\lst_s}}  & (\text{Partners Match}) \\
                           - & \existmatchjudge{\theta\bar{H}_p}{\lst} \bgap \existmatchjudge{\theta\bar{H}_s}{\lst} 
                               \bgap \existmatchjudge{\theta C}{\lst} \sgap &  (\text{Maximal Comprehension}) 
                         \end{array}  
                       \end{array} \\ \hline

   (\mt{act}\tx{-}\mt{simpa}\tx{-}2) &  \begin{array}{l}
                         
                         \begin{array}{l}
                           \wprog ~ \rhd \execstate{\stack{\actgoal~\idcons{A}{n}~i}{\gs}}
                                                    {\mset{\lst \msetcomma \lst_p \msetcomma \lst_s \msetcomma \lst_a \msetcomma \idcons{A}{n}}} \\
                           \optrans ~ \execstate{\sstack{\initgoal~\theta \bar{B}} \stplus \stack{\actgoal~\idcons{A}{n}~i}{\gs}}
                                              {\mset{\lst \msetcomma \lst_p \msetcomma \lst_a \msetcomma \idcons{A}{n}}} 
                         \end{array} \\
                          \text{if}~\lkprog{\wprog}{i} = (\chrrule{r}{\mset{\bar{H}_p \msetcomma \occ{C}{i}}}{\bar{H}_s}{g}{\bar{B}})
                         ~\text{and}~\bar{H}_s \neq \msetempty,~\text{there exists some}~\theta~\text{such that}~ \\
                         \sgap \begin{array}{lll}
                           - & \guardjudge{\theta g} \bgap \matchjudge{\theta C}{\mset{\drop{\lst_a} \msetcomma \idcons{A}{n}}} & (\text{Guard Satisfied and Active Match}) \\
                           - & \matchjudge{\theta\bar{H}_p}{\drop{\lst_p}} \bgap \matchjudge{\theta\bar{H}_s}{\drop{\lst_s}}  & (\text{Partners Match}) \\
                           - & \existmatchjudge{\theta\bar{H}_p}{\lst} \bgap \existmatchjudge{\theta\bar{H}_s}{\lst} 
                               \bgap \existmatchjudge{\theta C}{\lst} \sgap &  (\text{Maximal Comprehension}) 
                         \end{array}  
                       \end{array} \\ \hline

   (\mt{act}\tx{-}\mt{next}) &  \begin{array}{l}
                                   \wprog \rhd \execstate{\stack{\actgoal~\idcons{A}{n}~i}{\gs}}{\lst}
                                  \optrans \execstate{\stack{\actgoal~\idcons{A}{n}~(i+1)}{\gs}}{\lst} \\
                                   \text{if neither}~(\mt{act}\tx{-}\mt{simpa}\tx{-}1)~\text{nor}~(\mt{act}\tx{-}\mt{simpa}\tx{-}2)~\text{applies}.
                                 \end{array} \\ \hline

   (\mt{act}\tx{-}\mt{drop}) &  \wprog \rhd \execstate{\stack{\actgoal~\idcons{A}{n}~i}{\gs}}{\lst} \optrans \execstate{\gs}{\lst}
                                    \bgap \bgap \text{if}~\lkprog{\wprog}{i} = \bot \\ \hline
\end{array} 
\]
\vspace{-6mm}
}}

The operational semantics of $\chrcp$ is defined by the judgment
$\ruleoptrans{\wprog}{\sigma}{\sigma'}$, where $\wprog$ is an annotated
$\chrcp$ program and $\sigma$, $\sigma'$ are execution states. It describes
the goal-orientated execution of the $\chrcp$ program $\wprog$. We write
$\progoptrans{\wprog}{\sigma}{\sigma'}$ for zero or more such derivation
steps.  The operational semantics introduces administrative derivation steps
that describe the incremental processing of constraints, as well as the
saturation of propagation rule applications (see below).  Execution starts
from an \emph{initial} execution state $\sigma$ of the form
$\execstate{\sstack{\initgoal~\bar{B}}}{\msetempty}$ where $\bar{B}$ is the
initial multiset of constraints.  Figure~\ref{fig:op_sem_core} shows the core
set of rules for this judgment. They handle all cases except those for
propagation rules. Rule ($\mt{init}$) applies when the leading goal has the
form $\initgoal~\bar{B}$.  It partitions $\bar{B}$ into $\bar{B}_l$ and
$\bar{B}_e$, both of which are unfolded into $\st_l$ and $\st_e$ respectively.
$\bar{B}_l$ contains the multiset of constraints which are monotone w.r.t. to
$\wprog$ (i.e., $\existunfprogjudge{\wprog}{\bar{B}_l}$).  These constraints
are \emph{not} added to the store immediately, rather we only add them into
the goal as `$\lazygoal$` goals (lazily stored).  Constraints $\bar{B}_e$ are
not monotone w.r.t. to $\wprog$, hence they are immediately added to the store
and added to the goals as `$\eagergoal$' goals (eagerly stored). This is key
to preserving the soundness of the operational semantics w.r.t.\@ the abstract
semantics.  Rule ($\mt{lazy}\tx{-}\mt{act}$) handles goals of the form
$\lazygoal~A$: we initiate active matching on $A$ by adding it to the store
and adding the new goal $\actgoal~\idcons{A}{n}~1$.  Rules
($\mt{eager}\tx{-}\mt{act}$) and ($\mt{eager}\tx{-}\mt{drop}$) deal with the
cases of goals of the form $\eagergoal~\idcons{A}{n}$. The former adds the
goal `$\actgoal~\idcons{A}{n}~1$' if $\idcons{A}{n}$ is still present in the
store, while the later simply drops the leading goal otherwise. The last four
rules handle case for a leading goal of the form $\actgoal~\idcons{A}{n}~i$:
rules ($\mt{act}\tx{-}\mt{simpa}\tx{-}1$) and
($\mt{act}\tx{-}\mt{simpa}\tx{-}2$) handle the cases where the active
constraint $\idcons{A}{n}$ matches the $i^{th}$ rule head occurrence of
$\wprog$, which is a simplified or propagated head respectively. If this match
satisfies the rule guard condition, matching partners exist in the store and
the comprehension maximality condition is satisfied, we apply the
corresponding rule instance. To define these matching conditions, we use the
auxiliary judgments defined by the abstract semantics
(Figure~\ref{fig:rule_app_judgment}).  Note that the rule body instance
$\theta \bar{B}$ is added as the new goal $\initgoal~\bar{B}$.  This is
because it potentially contains non-monotone constraints: we will employ rule
($\mt{init}$) to determine the storage policy of each constraint. For rule
($\mt{act}\tx{-}\mt{simpa}\tx{-}2$), we have the additional condition that the
simplified head of the rule be not empty, hence this case does not apply for
propagation rules. Rule ($\mt{act}\tx{-}\mt{next}$) applies when the previous
two rules do not, hence we cannot apply any instance of the rule with
$\idcons{A}{n}$ matching the $i^{th}$ rule head.  Finally, rule
($\mt{act}\tx{-}\mt{drop}$) drops the leading goal if occurrence index $i$
does not exist in $\wprog$. Since the occurrence index is incremented by
($\mt{act}\tx{-}\mt{next}$) starting with the activation of the goal at index
$1$, this indicates that we have exhaustively matched the constraint
$\idcons{A}{n}$ against all rule head occurrences.

\fig{fig:op_sem_prop}{Operational Semantics of $\chrcp$ (Propagation-Set)}{
{\small
\[
\hspace{-10mm}
\arraycolsep=1.4pt\def\arraystretch{1.2}
\begin{array}{|c|l|}
   \hline
   (\mt{act}\tx{-}\mt{prop}) & \begin{array}{l}
                                 \prog \rhd \execstate{\stack{\actgoal~\idcons{A}{n}~i}{\gs}}{\lst} 
                                            \optrans \execstate{\stack{\propgoal~\idcons{A}{n}~i~\msetempty}{\gs}}{\lst} \\
                                 \text{if}~\lkprog{\wprog}{i} = (\chrrule{r}{\bar{H}_p}{\msetempty}{g}{\bar{B}}) 
                               \end{array} \\ \hline

   (\mt{prop}\mt{-}\mt{prop}) & \begin{array}{l}
                         
                         \begin{array}{l}
                           \wprog ~ \rhd \execstate{\stack{\propgoal~\idcons{A}{n}~i~\Theta}{\gs}}
                                                    {\mset{\lst \msetcomma \lst_p \msetcomma \lst_a \msetcomma \idcons{A}{n}}} \\
                           \optrans ~ \execstate{\sstack{\initgoal~\theta \bar{B}} \stplus \stack{\propgoal~\idcons{A}{n}~i~(\Theta \cup (\theta,\bar{n}))}{\gs}}
                                                {\mset{\lst \msetcomma \lst_p \msetcomma \lst_a \msetcomma \idcons{A}{n}}} 
                         \end{array} \\
                          \text{if}~\lkprog{\wprog}{i} = (\chrrule{r}{\mset{\bar{H}_p \msetcomma \occ{C}{i}}}{\msetempty}{g}{\bar{B}}),
                         ~\text{there exists some}~\theta~\text{such that}~ \\
                          \sgap \begin{array}{lll}
                           - & \bar{n} \equiv \getlabels{\mset{\lst_p \msetcomma \lst_a \msetcomma \idcons{A}{n}}} \bgap 
                               (\theta,\bar{n}) \notin \Theta & (\text{Unique Instance}) \\
                           - & \guardjudge{\theta g} \bgap \matchjudge{\theta C}{\mset{\drop{\lst_a} \msetcomma \idcons{A}{n}}} & (\text{Guard Satisfied and Active Match}) \\
                           - & \matchjudge{\theta\bar{H}_p}{\drop{\lst_p}}  & (\text{Partners Match}) \\
                           - & \existmatchjudge{\theta\bar{H}_p}{\lst} \bgap \existmatchjudge{\theta C}{\lst} &  (\text{Maximal Comprehension}) 
                         \end{array}  
                       \end{array} \\ \hline

   (\mt{prop}\tx{-}\mt{sat}) & \begin{array}{l}
                                 \prog \rhd \execstate{\stack{\propgoal~\idcons{A}{n}~i~\Theta}{\gs}}{\lst}
                                  \optrans \execstate{\stack{\actgoal~\idcons{A}{n}~(i+1)}{\gs}}{\lst} \\
                                 \text{if}~(\mt{prop}\tx{-}\mt{prop})~\text{does not apply}.
                               \end{array} \\ \hline 
\end{array}
\]
\vspace{-6mm}
}}

Figure~\ref{fig:op_sem_prop} defines the rules that handle propagation rules.
Propagation rules need to be managed specially to avoid non-termination. Rule
($\mt{act}\tx{-}\mt{prop}$) defines the case where the active goal
$\actgoal~\idcons{A}{n}~i$ is such that the rule head occurrence index $i$ is
found in a propagation rule, then we replace the leading goal with
$\propgoal~\idcons{A}{n}~i~\msetempty$. Rule ($\mt{prop}\tx{-}\mt{prop}$)
applies an instance of this propagation rule that has not been applied before:
the application history is tracked by $\Theta$, which contains a set of pairs
$(\theta,\bar{n})$.  Finally, ($\mt{prop}\tx{-}\mt{sat}$) handles the case
where ($\mt{prop}\tx{-}\mt{prop}$) no longer applies, hence \emph{saturation}
has been achieved.  Since we uniquely identify an instance of the propagation
rule by the pair $(\theta,\bar{n})$, saturation is based on unique
permutations of constraints in the store that match the rule heads.

\mysection{Correspondence with the Abstract Semantics}
\label{sec:sound}

\fig{fig:corresp}{Correspondence Relation}{
\[
\hspace{-10mm}
\begin{array}{c}
 \begin{array}{ccccc}
 \text{Multisets} & \begin{cases} 
                      \sgap \corresp{\mset{\bar{o} \msetcomma o}} \sgap & \Coloneqq \sgap \mset{\corresp{\bar{o}} \msetcomma \corresp{o}} \\
                      \sgap \corresp{\msetempty} \sgap & \Coloneqq \sgap \msetempty
                    \end{cases} 
 & \bgap & 
 \text{Rule~Head} & \begin{cases} \sgap \corresp{~\occ{C}{i}~} \sgap \Coloneqq \sgap C \end{cases}
 \end{array} \vspace{2mm} \\

 \begin{array}{rl}
   \text{Rule} & \begin{cases} 
                    \sgap \corresp{~\chrrule{r}{\bar{H}_p}{\bar{H}_s}{g}{\bar{B}}~} 
                    \sgap \Coloneqq \sgap \chrrule{r}{\corresp{\bar{H}_p}}{\corresp{\bar{H}_s}}{g}{\bar{B}} 
                 \end{cases}
 \end{array} \vspace{2mm} \\

 \begin{array}{ccc}
   \begin{array}{rl}
     \text{State}  & \begin{cases} 
                       \sgap \corresp{~\execstate{\gs}{\lst}~} \sgap \Coloneqq \sgap \mset{\corresp{\gs} \msetcomma \corresp{\lst}} 
                     \end{cases} \vspace{2mm} \\

    \text{Constraint} & \begin{cases} \sgap \corresp{~\idcons{A}{n}~} \sgap \Coloneqq \sgap A \end{cases} \vspace{2mm} \\

     \text{Goals}  & \begin{cases}
                       \sgap \corresp{\stack{G}{\gs}} \sgap & \Coloneqq \sgap \mset{\corresp{G} \msetcomma \corresp{\gs}} \\
                       \sgap \corresp{\stackempty} \sgap  & \Coloneqq \sgap \msetempty
                     \end{cases} 
   \end{array} & ~ 
   \begin{array}{rl}
    \text{Goal} & \begin{cases}
                   \sgap \corresp{~\initgoal~\bar{B}~} \sgap \Coloneqq \sgap \st \sgap \text{s.t.} \sgap \rhsjudge{\bar{B}}{\st} \\
                   \sgap \corresp{~\lazygoal~A~}       \sgap \Coloneqq \sgap \mset{A} \\
                   \sgap \corresp{~\eagergoal~\idcons{A}{n}~} \sgap \Coloneqq \sgap \msetempty \\
                   \sgap \corresp{~\actgoal~\idcons{A}{n}~i~} \sgap \Coloneqq \sgap \msetempty \\
                   \sgap \corresp{~\propgoal~\idcons{A}{n}~i~\Theta~} \sgap \Coloneqq \sgap \msetempty
                  \end{cases} 
   \end{array} 
 \end{array} 
\end{array}
\]
\vspace{-6mm}
}

In this section, we relate the operational semantics shown in
Section~\ref{sec:op_sem} with the abstract semantics
(Section~\ref{sec:abs_sem}). Specifically, we prove the soundness of the
operational semantics w.r.t. the abstract semantics.

Figure~\ref{fig:corresp} defines a correspondence relation between
meta-objects of the operational semantics and those of the abstract semantics.
Given an object $o_\omega$ of the operational semantics, $o_\alpha =
\corresp{o_\omega}$ is the corresponding abstract semantic object.  For
instance, $\corresp{\wprog}$ strips occurrence indices away from $\wprog$.
Instead, the abstract constraint store $\corresp{\execstate{\gs}{\lst}}$
contains constraints in $\lst$ with labels removed, and the multiset union of
constraints found in `$\initgoal$' and `$\lazygoal$' goals of $\gs$.

We also need to define several invariants and prove that they are preserved
throughout the derivations of the operational semantics of $\chrcp$. An
execution state $\execstate{\gs}{\lst}$ is \emph{valid} for program
$\wprog$ if: \vspace{-2mm}
\begin{itemize}
\item%
  $\existunfprogjudge{\prog}{A}$ for any $\lazygoal~A \in \gs$.
\item%
  If $\gs = \stack{G}{\gs'}$, then $\initgoal~\bar{B} \notin \gs'$ for any
  $\bar{B}$.
\end{itemize}
\vspace{-2mm}
\comment{
\[
\begin{array}{c}
  \execstate{\gs}{\lst} ~\text{is \emph{initial}} \sgap \mt{iff} \sgap 
  \text{for some}~\bar{B} \sgap \gs = \sstack{\initgoal~\bar{B}} \sgap \lst = \msetempty \\ \\

  \execstate{\gs}{\lst} ~\text{is \emph{terminal}} \sgap \mt{iff} \sgap
  \gs = \epsilon \\ \\

  \execstate{\gs}{\lst} ~\text{is \emph{valid} w.r.t}~\prog \sgap \mt{iff} \sgap
  \begin{cases}
    \sgap \text{For any} \sgap \lazygoal~A \in \gs \sgap \text{we have} \sgap \existunfprogjudge{\prog}{A} \\
    \sgap \text{For any} \sgap \actgoal~\idcons{A}{n}~i \in \gs  \sgap \text{we have} \sgap \idcons{A}{n} \in \lst \\
    \sgap \text{For any} \sgap \propgoal~\idcons{A}{n}~i~\Theta \in \gs \sgap \text{we have} \sgap \idcons{A}{n} \in \lst \\
    \sgap \text{If} \sgap \gs = \stack{g}{\gs'} \sgap \text{then for any}~\bar{B}~\text{we have} \sgap \initgoal~\bar{B} \notin \gs'
  \end{cases}
\end{array}
\]
}
Initial states of the form $\execstate{\sstack{\initgoal~\bar{B}}}{\msetempty}$ are trivially
valid states. Lemma~\ref{lem:preserve} proves that derivation steps $\wprog \rhd \sigma \optrans \sigma'$
preserve validity during execution.
\vspace{-2mm}
\begin{lemma}[Preservation]
\label{lem:preserve}%
  For any program $\wprog$, given any valid state $\sigma$ and any state $\sigma'$, if
  $\wprog \rhd \sigma \optrans \sigma'$, then $\sigma'$ must be a valid state.
\end{lemma}
\vspace{-2mm}
\begin{proof}
  The proof proceeds by structural induction on all possible forms of
  derivation $\wprog \rhd \sigma \optrans \sigma'$.  It is easy to show that
  each transition preserves validity.
\qed%
\end{proof}

Lemma~\ref{lem:corr_step} states that any derivation step of the operational
semantics $\wprog \rhd \sigma \optrans \sigma'$ is either silent in the
abstract semantics (i.e., $\corresp{\sigma} \equiv \corresp{\sigma'}$) or
corresponds to a valid derivation step (i.e., $\corresp{\wprog} \rhd
\corresp{\sigma} \abstrans \corresp{\sigma'}$).

\vspace{-2mm}
\begin{lemma}[Correspondence Step]
\label{lem:corr_step}%
  For any program $\wprog$ and valid execution states $\sigma$, $\sigma'$, if $\wprog \rhd \sigma \optrans \sigma'$, then
  either $\corresp{\sigma} \equiv \corresp{\sigma'}$ or $\corresp{\wprog} \rhd \corresp{\sigma} \abstrans \corresp{\sigma'}$.
\end{lemma}
\vspace{-2mm}
\begin{proof}
  The proof proceeds by structural induction on all possible forms of
  derivation $\wprog \rhd \sigma \optrans \sigma'$.  Rules
  ($\mt{act}\tx{-}\mt{simpa}\tx{-}1$), ($\mt{act}\tx{-}\mt{simpa}\tx{-}2$) and
  ($\mt{prop}\tx{-}\mt{prop}$) correspond to abstract steps.  For them, we
  exploit conditional monotonicity in Theorem~\ref{theo:cond_mono} and
  preservation in Lemma~\ref{lem:preserve} to guarantee the validity of
  corresponding derivation step in the abstract semantics. All other rules are
  silent.  \qed%
\end{proof}

\vspace{-3mm}
\begin{theorem}[Soundness]
\label{theo:sound}%
For any program $\wprog$ and valid execution states $\sigma$ and $\sigma'$, if
$\wprog \rhd \sigma \optrans^* \sigma'$, then $\corresp{\wprog} \rhd
\corresp{\sigma} \abstrans^* \corresp{\sigma'}$.
\end{theorem}
\vspace{-3mm}
\begin{proof}
  The proof proceeds by induction on derivation steps.  The inductive case is
  proved using Lemmas~\ref{lem:corr_step} and~\ref{lem:preserve}.
\qed%
\end{proof}

While the operational semantics is sound, completeness w.r.t. the abstract
semantics does not hold.  There are two aspects of the operational semantics
that contributes to this: first the saturation behavior of propagation rules
(Figure~\ref{fig:op_sem_prop}) is not modeled in the abstract semantics. This
means that while a program $\wprog$ with a propagation rule terminates in the
operational semantics (thanks to saturation), $\corresp{\wprog}$ may diverge
in the abstract semantics. Second, although we can model negation with
comprehension patterns, we cannot guarantee completeness when we do so. For
instance, consider the rule $\mt{\ngmcompre{a(X)}{X \unidis Xs}
  \Longleftrightarrow Xs = \msetempty \mid noA}$, which adds $noA$ to the
constraint store if there are no occurrences of $\mt{a(X)}$ for any value of
$X$. The application of this rule solely depends on the absence of any
occurrences $\mt{a(X)}$ in the store.  Yet, in our operational semantics, rule
application is triggered only by the presence of constraints. The idea of
negated active constraint can be borrowed from~\cite{Weert06extendingchr} to
rectify this incompleteness, but space limitations prevent us from discussing
the details of this conservative extension to our operational semantics.

\mysection{Related Work}
\label{sec:related}

An extension of CHR with aggregates is proposed
in~\cite{DBLP:conf/iclp/SneyersWSD07}.  This extension allows the programmer
to write CHR rules with aggregate constraints that incrementally maintains
term-level aggregate computations. Differently from our comprehension
patterns, these aggregate constraints are only allowed to appear as
\emph{propagated} rule heads.  The authors of
\cite{DBLP:conf/iclp/SneyersWSD07} also suggested extending the refined CHR
operational semantics~\cite{Duck04therefined} with aggregates, in a manner
analogous to their previous work on CHR with negated
constraints~\cite{Weert06extendingchr}.  While both extensions (aggregates and
negated constraints) introduce non-monotonicity in the respective CHR
semantics, the observable symptoms (from an end-user's perspective) of this
non-monotonicity are described as ``unexpected behaviors''
in~\cite{Weert06extendingchr}, serving only as caveats for the programmers. No
clear solution is proposed at the level of the semantics.  By contrast, our work
here directly addresses the issue of incrementally processing of constraints
in the presence of non-monotonicity introduced by comprehension patterns.

The logic programming language Meld~\cite{cruz-damp12-short}, originally
designed to program cyber-physical systems, offers a rich range of features
including aggregates and a limited form of comprehension patterns.  To the
best of our knowledge, a low-level semantics on which to base an efficient
implementation of Meld has not yet been explored.  By contrast, our work
explores comprehension patterns in multiset rewriting rules in detail and
defines an operational semantics that is amenable to an incremental strategy
for processing constraints.

\mysection{Conclusion and Future Works}
\label{sec:conclude}

In this paper, we introduced $\chrcp$, an extension of CHR with multiset
comprehension patterns.  We defined an abstract semantics for $\chrcp$,
followed by an operational semantics and proved its soundness with respect to
the abstract semantics. We are currently developing a prototype implementation
based on the operational semantics discussed here.

In future work, we intend to further develop our prototype implementation of
$\chrcp$ by investigating efficient compilation schemes for comprehension
patterns. We also wish to explore alternatives to our current greedy
comprehension pattern matching scheme.  We also intend to extend $\chrcp$ with
some result form prior work in~\cite{ppdp13} and develop a
decentralized multiset rewriting language.  Finally, we believe that the
operational semantics of $\chrcp$ can be generalized to other non-monotonic
features along the lines
of~\cite{DBLP:conf/iclp/SneyersWSD07,Weert06extendingchr}.

\vspace{-4mm}
\bibliography{references}

\end{document}